\documentclass[letterpaper, 10 pt, conference]{ieeeconf}  

\usepackage{color,array}

\usepackage{amssymb,amsfonts}

\usepackage{graphicx}
\usepackage{verbatim}   

\usepackage{algorithm}
\usepackage{algpseudocode}

\usepackage{mathtools}
\usepackage{siunitx}
\usepackage{tabularx}
\usepackage{pgfplots}
\usepackage{pgfplotstable}

\usepackage{multirow}
\usepgfplotslibrary{fillbetween}
\usepackage{booktabs}
\pgfplotsset{compat = newest}
\usetikzlibrary{arrows,positioning,shapes,intersections,external,patterns,calc,fit,decorations,decorations.markings} 
\usepackage{cite}
\newcommand{\R}{\mathbb{R}}

\newcommand{\J}{\mathbb{J}}

\newcommand{\SO}{\text{SO}}

\newcommand{\eff}{\text{eff}}

\newtheorem{definition}{Definition}
\newtheorem{lemma}{Lemma}

\newtheorem{proposition}{Proposition}[section]
\newtheorem{remarkth}[definition]{Remark}
\newenvironment{remark}{\begin{remarkth}\upshape}{\end{remarkth}}

\IEEEoverridecommandlockouts                              

\title{\LARGE \bf
Observer-Based Estimation and Hydrostatic Inertia Modeling for Cooperative Transport of Variable-Inertia Loads with Quadrotors
}

\author{Jacob R. Goodman, Leonardo J. Colombo, Juan I. Giribet
\thanks{J. Goodman is with Department of Mathematical Sciences NTNU, Norwegian University of Science and Technology, Norway.{\tt\small jacob.goodman@ntnu.no.}}
\thanks{Juan I. Giribet is with Universidad de San Andr\'es (UdeSA) and CONICET, Argentina.
        {\tt\footnotesize jgiribet@conicet.gov.ar}}     
        \thanks{L. Colombo is with Centre for Automation and Robotics (CSIC-UPM), Ctra. M300 Campo Real, Km 0,200, Arganda
del Rey - 28500 Madrid, Spain.{\tt\small leonardo.colombo@csic.es}}%
\thanks{J. Goodman was supported by the Marie Skłodowska-Curie grant agreement No. 101206748 (GNACS). J. Giribet was supported by PICT-2019-2371 and PICT-2019-0373 projects from Agencia Nacional de Investigaciones Cient\'ificas y Tecnol\'ogicas, and UBACyT-0421BA project from the Universidad de Buenos Aires (UBA), Argentina. The research leading to these results was supported in part by iRoboCity2030-CM, Robótica Inteligente para Ciudades Sostenibles (TEC-2024/TEC-62), funded by the Programas de Actividades I+D en Tecnologías en la Comunidad de Madrid. The authors also acknowledge financial support from Grant PID2022-137909NB-C21 funded by MCIN/AEI/ 10.13039/501100011033}
}
\begin{document}

\maketitle
\thispagestyle{empty}
\pagestyle{empty}


\begin{abstract}
We address load-parameter estimation in cooperative aerial transport with time-varying mass and inertia, as in fluid-carrying payloads. Using an intrinsic manifold model of the multi–quadrotor–load dynamics, we combine a geometric tracking controller with an observer for parameter identification. We estimate mass from measurable kinematics and commanded forces, and handle variable inertia via an inertia surrogate that reproduces the load’s rotational dynamics for control and state propagation.

Instead of real-time identification of the true inertia tensor, driven by high-dimensional internal fluid motion, we leverage known tank geometry and fluid-mechanical structure to precompute inertia tensors and update them through a lookup table indexed by fill level and attitude. The surrogate is justified via the incompressible Navier–Stokes equations in the translating/rotating load frame: when effective forcing is gravity-dominated (i.e., translational/rotational accelerations and especially jerk are limited), the fluid approaches hydrostatic equilibrium and the free surface is well approximated by a plane orthogonal to the body-frame gravity direction. 
\end{abstract}

\section{Introduction}

Geometric control is a powerful framework for aerial robotics, handling the coupled translational and rotational dynamics of underactuated multirotors and payloads \cite{Elastic,Lee-TCST,LeeSrePICDC13}. By formulating dynamics and feedback directly on $\mathrm{SE}(3)$, it yields (almost) globally valid, structure-preserving laws and avoids the singularities of local attitude coordinates. Cooperative transportation with multiple quadrotors has therefore attracted growing interest for logistics, construction, and search-and-rescue \cite{Fink,loiano}, with recent experiments demonstrating agile manipulation of cable-suspended loads in cluttered environments \cite{sun2025agilecoop}.

In aerial transport, cables provide a lightweight coupling between UAVs and cargo \cite{pacheco1,pacheco2}. Geometric controllers have been developed for teams transporting suspended point-mass loads \cite{SreLeePICDC13}, with extensions to rigid-body payloads \cite{Lee-TCST,wu}. Although many works assume inelastic tethers, elasticity can be crucial for accurate force transmission and reliable state/parameter estimation; elastic-tether models and controllers were studied for single-UAV transport \cite{Elastic,wu} and cooperative rigid-load transport \cite{jacob1,jacob2}. Related estimation works include inertial identification in cooperative transport \cite{petitti2020inertial} and the impact of internal forces in cable co-manipulation \cite{tognon2018internalforce}.

A key limitation is that most available models/controllers assume constant load mass and inertia. This can fail for fluid-filled, moving, or deformable payloads, where the mass distribution evolves and inertial parameters vary in time, degrading closed-loop performance. Motivated by such scenarios, we consider a team of quadrotor UAVs transporting a rigid load with time-varying mass and inertia, suspended through elastic cables. To the best of our knowledge, this is the first work combining geometric modeling with controller and observer design in the presence of \emph{time-varying load inertia} for cooperative aerial transportation with elastic tethers.

The work in \cite{jacob1} derived a complete geometric variational model for a multi-quadrotor--load system with fixed load mass and inertia. Here we address the estimation challenge posed by \emph{unknown, time-varying} inertial parameters by developing online adaptation mechanisms while preserving the geometric structure of the closed-loop dynamics.

Main contributions: (i) we derive structure-preserving reduced equations for a cooperative multi-quadrotor--load system with time-varying load mass $m_L(t)$ and inertia tensor $\J_L(t)$, where the quadrotors are connected to the load through elastic tethers; (ii) leveraging the reduced dynamics and available measurements, we propose an observer to estimate the time-varying load mass and characterize its error dynamics under bounded disturbances and model mismatch; (iii) to cope with inertia variations induced by internal fluid motion, we introduce a computationally efficient hydrostatic surrogate in which $\J_L(t)$ is updated via a look-up table precomputed from known tank geometry and a hydrostatic approximation motivated by the body-frame Navier--Stokes equations, with the look-up indexed by the estimated mass and the measured load attitude.

The remainder of the paper is organized as follows. Section~II presents the geometric model and derives reduced equations with time-varying load mass and inertia. Section~III develops the estimation framework, including the online mass observer and hydrostatic inertia update via offline precomputation. Section~IV reports simulation results, and Section~V concludes and outlines future work.

\section{Modelization and Control Design}\label{sec3}
We begin by modelling and deriving the control system describing the transportation task between the quadrotors. This is done by constructing the total kinetic and potential energies of the system, in addition to the virtual work done by non-conservative forces and, subsequently, by using tools from variational calculus on manifolds. 

Consider $N$ identical quadrotor UAVs transporting a rigid body of variable total mass $m_{L}: [0, T] \to \in\mathbb{R}_{>0}$ and variable positive-definite inertia matrix $\mathbb{J}_{L}(t): [0, T] \to \mathbb{R}^{3\times 3}$. It is connected to the center of mass of each quadrotor via a massless inflexible cable of length $L$. 


The configuration space of the mechanical system that describes the cooperative task is given by $$Q = \underbrace{(\SO(3) \times \R^{3})}_{\text{Rigid load}} \times\underbrace{(S^{2})^{N}}_{\text{Cables}}\times\underbrace{(\SO(3))^{N}}_{\text{Quadrotor attitudes}}.$$ The system has $5N+1$ degrees of freedom: $6$ degrees corresponding to the load, $2$ degrees for each cable, and $3$ degrees for each quadrotor. Note that the positions of the quadrotors do not appear in the configuration space, as they are uniquely defined in terms of the other state variables via the physical relation $x_{Q_{j}} = x_{L} + R_{L} r_{j} - L q_{j}$ for $j \in \{1, \dots, N\}$. Meanwhile, we have $4N$ inputs to the system in the form of $N$ thrust controls $f_{j} \in \R$, corresponding to the total lift force exerted on the quadrotors by the spinning propellers, and $N$ moment controls $M_{j} \in \R^{3}$, which are related to the torques induced on the quadrotors by the rotating propellers, for each $j \in \{1, \dots, N\}$. Hence, the complete system has $2N + 6$ degrees of under-actuation, with the quadrotor positions and attitudes being directly actuated. Note also that upon fixing a body frame to each quadrotor such that the vector $\bar{e}_{3} = [0, 0, 1]^{T}$ points in the direction of the applied thrust, we may alternatively express the thrust controls $f_{j}$ via the vectors $u_{j}=f_{j}R_{j}e_{3}\in\mathbb{R}^{3}$ in the inertial frame for $j \in \{1,\dots, N\}$. We will utilize this representation frequently throughout the section. Alternatively, one may choose to control the total thrust of \textit{each} propeller individually. We adopt the former approach since it is standard in the literature and cleanly separates the quadrotor attitude dynamics from the remaining dynamics. The thrust of the $i^{\text{th}}$ propeller along the $e_{3}$-axis follows from the total thrust and moment controller \cite{lee_geometric_2010}.

This thrust/torque representation allows us to separate the quadrotor attitude dynamics from the
translational interaction forces and to write the closed-form reduced equations that will later be
used for estimation. Substituting the constraint $x_{Q_j}=x_L+R_L r_j-Lq_j$ and applying the
Lagrange--d'Alembert principle on $Q$, we obtain the following reduced controlled
Euler--Lagrange equations. When $m_L$ and $\mathbb{J}_L$ are constant, they coincide (up to notation)
with the standard cooperative transport models in \cite{jacob1}; here the time
variation of $m_L(t)$ and $\mathbb{J}_L(t)$ yields the additional terms $\dot m_L v_L$ and
$\dot{\mathbb{J}}_L\Omega_L$:

\begin{equation}
\dot{x}_L = v_L,\quad\dot{R}_L = R_L \hat{\Omega}_L,\quad \dot{q}_j = \omega_j \times q_j, \label{xred}
\end{equation}

\begin{align} 
&m_{\eff} ( \dot{v}_L + ge_3) + \dot{m}_L v_L = \nonumber\\
&\hspace{1cm}\sum_{j=1}^N \Big(u_j - m_Q R_L(\hat{\Omega}_L^2 + \dot{\hat{\Omega}}_L)r_j + m_Q L\ddot{q}_j \Big), \label{vred} 
\end{align}
\begin{align}
&\mathbb{J}_{\eff} \dot{\Omega}_L + \hat{\Omega}_L \mathbb{J}_{\eff} \Omega_L + \dot{\mathbb{J}}_L \Omega_L =\nonumber\\ &\hspace{1cm}\sum_{j=1}^N m_Q \hat{r}_j R_L^T \Big(-g e_3 - \dot{v}_L + L\ddot{q}_j + \frac1{m_Q}u_j\Big),\label{OmegaLred}\\ 
&\dot{\omega}_j = \frac1{L} \hat{q}_j\Big( \dot{v}_L + R_L(\hat{\Omega}_L^2 + \dot{\hat{\Omega}}_L)r_j + ge_3 - \frac1{m_Q} u_j \Big), \label{wjred}\\
&\mathbb{J}_Q \dot{\Omega}_j = \mathbb{J}_Q \Omega_j \times \Omega_j + M_j,\,\, \dot{R}_j = R_j \hat{\Omega}_j,\label{Rdynred}.
\end{align} where $j = 1,...,N$, $m_{\eff} := Nm_Q + m_L(t)$ and $\mathbb{J}_{\eff} := \mathbb{J}_L(t) - \sum_{j=1}^N m_Q \hat{r}_j^2$.

\begin{remark}
Equations \eqref{xred} describe the kinematics of the load position, load attitude, and the cable attitudes, respectively.  Equations \eqref{vred}-\eqref{OmegaLred} describe the dynamics of the load's position and attitude, respectively. Equation \eqref{wjred} describes the dynamics of each cable. Finally, equation \eqref{Rdynred} corresponds with the dynamics of the attitudes of each quadrotor.
\end{remark}
\begin{remark}
Note that the configuration space provides $5N + 1$ degrees of freedom and the controls provide $4N$ inputs to the system. Hence, the system has $2N + 6$ degrees of underactuation, and can be shown not to be differentially flat. That is, $f_1, \dots, f_N$ and $M_1, \dots, M_N$ cannot be designed to control the complete system. By avoiding certain degrees of freedom, one can make the system differentiable flat, as in the cooperative task with inelastic cables studied in \cite{wu}. 
\end{remark}

In what follows, we apply some feedback terms to the reduced model \eqref{xred}-\eqref{Rdynred} through the controls, allowing us to reach a simplified set of equations. We decompose $u_j$ into components which are parallel and perpendicular to the cable attitudes $q_j$ via $u_j = u_j^{\parallel} + u_j^{\perp}$, where $u_j^{\parallel} = (q_j^T u_j)q_j$ and $u_j^{\perp} = (I - q_j q_j^T)u_j$. This is motivated by the fact that only the $u_j^{\parallel}$ components influence the dynamics of the load, while only the $u_j^{\perp}$ components influence the cables dynamics.


Our first step is to find an equation for $\ddot{q}_j$ that we will substitute into equations \eqref{vred} and \eqref{OmegaLred}. By differentiating $\dot{q}_j = \omega_j \times q_j$ and expanding it with the vector triple product identity, it can be shown that $\ddot{q}_j = \dot{\omega}_j \times q_j - \|\omega_j \|^2 q_j$. Now we may substitute \eqref{wjred} in for $\dot{\omega}_j$ to find that

\begin{align*}
L \ddot{q}_j &= -q_j \times (L\dot{\omega}_j) - L\|\omega_j\|^2 q_j \\
&= (I - q_j q_j^T)\Big[\dot{v}_L - R_L(\hat{\Omega}_L^2 + \dot{\hat{\Omega}}_L)r_j + ge_3 - \frac1{m_Q} u_j \Big] \\
&\quad -L\left\|\omega_j\right\|^2 q_j.
\end{align*} 
Substituting this equation for $m_Q L \ddot{q}_j$ into \eqref{vred} and making use of the fact that $m_{\eff} = Nm_Q + m_L(t)$, we obtain
\begin{align}
M_L (\dot{v}_L + ge_3) + \dot{m}_L v_L = &\sum_{j=1}^N \left[m_Qq_j q_j^T R_L (\hat{\Omega}_L^2 + \dot{\hat{\Omega}}_L)r_j\right.\nonumber\\ &+\left. u_j^{\parallel} - m_Q L \left\|\omega_j\right\|^2 q_j \right],\label{xLdyn_Lee}
\end{align}where $M_L = m_L(t) I + \sum_{j=1}^N m_Q q_j q_j^T$. 

Repeating this procedure with \eqref{OmegaLred}, and making use of the fact that $\mathbb{J}_{\eff} := \mathbb{J}_L(t) - \sum_{j=1}^N m_Q \hat{r}_j^2$ yields
\begin{equation}\label{OmegaLdyn_Lee}
\mathbb{J}_L \dot{\Omega}_L
+ \hat{\Omega}_L \mathbb{J}_L \Omega_L
+ \dot{\mathbb{J}}_L \Omega_L
= \sum_{j=1}^N m_Q \hat{r}_j R_L^T \,\Xi_j ,
\end{equation}
\begin{equation}
\begin{aligned}
\Xi_j :=\;&
q_j q_j^T R_L(\hat{r}_j \dot{\Omega}_L + \hat{\Omega}_L^2 r_j)
- q_j q_j^T(\dot{v}_L + g e_3) \\
&- L\|\omega_j\|^2 q_j
+ \frac{1}{m_Q}u_j^{\parallel}.
\end{aligned}
\end{equation}

We will now define our controls as \begin{align}
    u_j =& \mu_j + \nu_j + m_Q (\dot{v}_L + R_L(\hat{\Omega}^2_L + \dot{\hat{\Omega}}_L )r_j + ge_3)\label{eq: control}\\ &+ m_Q L \|\omega_j\|^2 q_j,\nonumber
\end{align} where $\mu_j, \nu_j \in \R^3$ are additional controls constrained to be parallel and perpendicular to the cable attitude $q_j$, respectively, for $j = 1, \dots, N.$ Observe that the load acceleration $\dot{v}_L$ and angular acceleration $\dot{\hat{\Omega}}$ appear in the control design \eqref{eq: control}. To avoid circular algebra, we note that these quantities will be replaced with measurements (on-board accelerometers/gyroscopes) or estimations in terms of measured quantities in practice.

With these feedback terms, it is easy to see that the system dynamics take the form:
\begin{align}
    m_L(\dot{v}_L + ge_3) + \dot{m}_L v_L &= \sum_{j=1}^4 \mu_j, \label{xdynred2}\\
    \dot{R}_L &= R_L \hat{\Omega}_L, \quad \dot{x}_L = v_L,  \label{RLkinred2}\\
    \mathbb{J}_L \dot{\Omega}_L + \hat{\Omega}_L \mathbb{J}_L \Omega_L + \dot{\mathbb{J}}_L \Omega_L &= \sum_{j=1}^4 \hat{r}_j R_L^T \mu_j, \label{RLdynred2}\\
  \dot{q}_j = \omega_j \times q_j,\quad   m_Q \dot{\omega}_j &= -q_j \times \nu_j
    \label{qdynred2}\\
  \dot{R}_j = R_j, \quad  \mathbb{J}_Q \dot{\Omega}_j &= \mathbb{J}_Q \Omega_j \times \Omega_j + M_j. \label{Rjdynred2}
\end{align}

\section{Estimation of Load Mass and Moment of Inertia}
For fluid-filled payloads (e.g., leaking or dischargeable tanks), both the mass $m_L(t)$ and inertia tensor $\J_L(t)$ can vary in time. While mass can often be measured with added hardware (e.g., a load cell), there is typically no direct, reliable way to measure the full inertia tensor online, since fluid reconfiguration depends on motion and is not directly sensed without much more complex instrumentation and modeling. In many deployments, online mass sensing is unavailable or biased by leaks, discharge uncertainty, or calibration errors. We therefore target variable-mass operation \emph{without} hardware modifications, developing an online mass estimator using standard state estimates and control inputs, complemented by an inertia update strategy based on an offline tank-geometry/fill-level model.

Throughout, we assume that the load state is available through onboard state estimation, in particular that we have access to estimates of
$x_L$, $v_L$, $\dot v_L$, $R_L$, and $\Omega_L$, as well as the control inputs $\mu_j$ for $j=1,\dots 4$.


\subsection{Mass Estimation}
We now seek to estimate the load mass $m_L$ under the assumption that we are given some parametric model of the mass $m_\theta$, for parameters $\theta \in \R^k$. In practice, there are several plausible mass models which cover an array of situations, such as:
\begin{enumerate}
    \item (\textit{Constant mass}) $\theta = m_0$, with $m_\theta(t) = m_0$ for all $t \in \R$. 
    \item (\textit{Orifice leak}) $\theta = (m_0, \lambda)$, with $m_\theta(t) = (\sqrt{m_0} - \sqrt\lambda t)^2$.
    \item (\textit{Viscous leak}) $\theta = (m_0, \lambda)$, with $m_\theta(t) = m_0 e^{-\lambda t}$.
\end{enumerate}
For fluid-carrying missions with subsequent dumping, a constant-mass model is often adequate, with the total mass treated as unknown and corrected online. The other two cases capture standard leakage mechanisms: an orifice leak governed by Torricelli’s law, and a viscous leak through a crack or narrow tube governed by the Hagen–Poiseuille law.

In what follows, we suppose that estimations of the state variables $x_L, v_L, \dot{v}_L, R_L,$ and $\Omega_L$, as well as the control inputs $\mu_j$, are given. Equation \eqref{xdynred2} can be interpreted as the ideal mass regression. Assuming that $m_L = m_\theta$ for some (constant) parameter which best fits our model, we have:
\begin{equation}\label{eq: exact_mass_dyn}
    \sum_{j=1}^4 \mu_j - m_\theta(\dot{v}_L + ge_3) - \dot{m}_\theta v_L + \Delta_m = 0
\end{equation}
where $\Delta_m$ can be viewed as an additive disturbance to our mass dynamics, corresponding to model-mismatch, or measurement errors in our state variables and accelerations. Let $\hat{\theta}(t) \in \R^k$ denote an estimate for our ideal parameters $\theta$, and $m_{\hat{\theta}}$ denote the corresponding mass estimate. Then, for each $t \ge 0$, we wish to minimize the following quadratic cost function:
$$J(\hat{\theta}) = \frac12 \left\|\sum_{j=1}^4 \mu_j - m_{\hat{\theta}}(\dot{v}_L + ge_3) - \frac{\partial m_{\hat{\theta}}}{\partial t} v_L \right\|^2$$
This motivates the choice $\displaystyle{\dot{\hat{\theta}} = -K \frac{\partial J}{\partial \theta}(\hat{\theta})}$ where $K= \text{diag}(\gamma_1, \dots, \gamma_k)$ with $\gamma_i > 0$ for all $1 \le i \le k$. That is, we update $\hat{\theta}(t)$ by following the instantaneous gradient descent path of the cost $J$ (with time $t$ measurements for state variables/controls frozen), where each parameter has the corresponding \textit{learning rate} $\gamma_i$. Letting $\displaystyle M_\theta = \frac{\partial m_\theta}{\partial t}$ and $w_L = \dot{v}_L + ge_3$, we find that
\small{\begin{equation}\label{eq: theta_regressor}
    \dot{\hat{\theta}}_i(t) = \gamma_i\left(\frac{\partial m_{\hat{\theta}}}{\partial \theta_i}w_L + \frac{\partial M_{\hat{\theta}}}{\partial \theta_i} v_L\right)^T\left(\sum_{j=1}^4 \mu_j - m_{\hat{\theta}}w_L - M_{\hat{\theta}}v_L \right)
\end{equation}}
\normalsize Define the estimation errors $e_m = m_{\hat{\theta}} - m_\theta$ and $e_M = M_{\hat{\theta}} - M_\theta$. Inputting \eqref{eq: exact_mass_dyn} into \eqref{eq: theta_regressor}, we find:
\begin{align*}
   \dot{\hat{\theta}}_i(t) = \gamma_i\left(\frac{\partial m_{\hat{\theta}}}{\partial \theta_i}w_L + \frac{\partial M_{\hat{\theta}}}{\partial \theta_i} v_L\right)^T \left(e_m w_L + e_M v_L - \Delta_m\right)
\end{align*} so that the error dynamics are given by:
\begin{equation*}
\begin{aligned}
\dot e_m &= e_M + \sum_{i=1}^k \frac{\partial m_{\hat\theta}}{\partial \theta_i}\,\dot{\hat\theta}_i = e_M - \sum_{i=1}^k \gamma_i\,\frac{\partial m_{\hat\theta}}{\partial \theta_i}\,
\psi_i^{\!T}\,\eta,\\[2mm]
\dot e_M &= \frac{\partial M_{\hat\theta}}{\partial t} - \frac{\partial M_{\theta}}{\partial t}
+ \frac{\partial M_{\hat\theta}}{\partial \theta}\,\dot{\hat\theta}= R_{\hat\theta} - \sum_{i=1}^k \gamma_i\,\frac{\partial M_{\hat\theta}}{\partial \theta_i}\,
\psi_i^{\!T}\,\eta,
\end{aligned}
\end{equation*}
\begin{equation*}
\psi_i := \frac{\partial m_{\hat\theta}}{\partial \theta_i}w_L
        + \frac{\partial M_{\hat\theta}}{\partial \theta_i}v_L,
\quad
\eta := e_m w_L + e_M v_L - \Delta_m .
\end{equation*} where $R_{\hat{\theta}} := \displaystyle\frac{\partial^2 m_{\hat{\theta}}}{\partial t^2} - \frac{\partial^2 m_\theta}{\partial t^2}$ is a curvature term of the model. Consider the $k \times 3$ matrix $Y$ with row vectors $\frac{\partial m_{\hat{\theta}}}{\partial \theta_i}w_L + \frac{\partial M_{\hat{\theta}}}{\partial \theta_i}v_L$. Letting $\nabla_\theta m_{\hat{\theta}}, \ \nabla_\theta M_{\hat{\theta}}$ denote the column vector of partial derivatives of $m_{\hat{\theta}}, \  M_{\hat{\theta}}$, respectively, it is easy to see that $Y = \nabla_\theta m_{\hat{\theta}} w_L^T + \nabla_\theta M_{\hat{\theta}} v_L^T$, so that we obtain the following condensed form for the error dynamics:
\begin{align*}
\dot e_m &= e_M - (\nabla_{\theta} m_{\hat{\theta}})^T K\, \Phi\, \eta,\quad\dot e_M = R_{\hat{\theta}} - (\nabla_{\theta} M_{\hat{\theta}})^T K\, \Phi\, \eta,
\end{align*}
where
$\Phi := \nabla_\theta m_{\hat{\theta}}\, w_L^T + \nabla_\theta M_{\hat{\theta}}\, v_L^T$.

Note that $K \succ 0$, so that it induces a positive definite quadratic form $\left<x, y\right>_K = x^T K y$ for all $x,y \in \R^k$, with $\|x\|_K^2 \ge \lambda_{\min}(K) \|x\|^2 = \displaystyle \left(\min_{1 \le i \le k} \gamma_i\right)\|x\|^2$ for all $x \in \R^k$. Letting $\xi = \begin{bmatrix}
    e_m & e_M
\end{bmatrix}^T$, we see that the error dynamics can be written as:
\begin{equation}\label{eq: error_compressed}
    \dot{\xi} = -S\xi + \Delta
\end{equation}
Where $S = (S_{ij})$ is the symmetric matrix with elements
\begin{align*}
    S_{11} &= \|\nabla_\theta m_{\hat{\theta}}\|_K^2 \|w_L\|^2 + \left<\nabla_\theta m_{\hat{\theta}}, \nabla_\theta M_{\hat{\theta}}\right>_K (w_L^T v_L) \\
    S_{12} &= \|\nabla_\theta m_{\hat{\theta}}\|^2_K (w_L^T v_L) + \left<\nabla_\theta m_{\hat{\theta}}, \nabla_\theta M_{\hat{\theta}}\right>_K \|v_L\|^2 \\
    & + \left<\nabla_\theta m_{\hat{\theta}}, \nabla_\theta M_{\hat{\theta}}\right>_K \|w_L\|^2 + \|\nabla_\theta M_{\hat{\theta}}\|_K^2 (w_L^T v_L) - 1 \\
    S_{22} &= \left<\nabla_\theta m_{\hat{\theta}}, \nabla_\theta M_{\hat{\theta}}\right>_K (w_L^T v_L) + \|\nabla_\theta M_{\hat{\theta}}\|_K^2 \|v_L\|^2
\end{align*}
and 
$$\Delta = \begin{bmatrix}
    -(\nabla_{\theta} m_{\hat{\theta}})^T K \left(\nabla_\theta m_{\hat{\theta}} w_L^T + \nabla_\theta M_{\hat{\theta}} v_L^T \right) \Delta_m \\
    R_{\hat{\theta}}-(\nabla_{\theta} M_{\hat{\theta}})^T K \left(\nabla_\theta m_{\hat{\theta}} w_L^T + \nabla_\theta M_{\hat{\theta}} v_L^T \right) \Delta_m
\end{bmatrix}$$
is a disturbance to the error dynamics. In order for $S$ to be positive semi-definite, it is necessary that the diagonal elements $S_{11}, S_{22} \ge 0$. However, it is clear that this assumption will not generally hold without further assumptions on the model and the system trajectories. This can be expressed in terms of \textit{persistent excitation} type condition:

\begin{lemma}\label{lemma: persistent_excitation}
    Let $\xi(t) \in C^1(\R, \R^n)$ satisfy $\dot{\xi} = -A\xi + D$ for some $A \in C^0(\R, \text{Sym}(n))$ and $D \in C^0(\R, \R^{n})$. Suppose that $\exists \mu, T, M > 0$ such that
    \begin{align}
        \int_t^{t+T} \lambda_{\min}(A(\tau))d\tau &\ge \mu \label{eq: persistent_ex}\\
        \int_t^{t+T} \lambda_{\min}(A(\tau))^+d\tau &\le M \label{eq: spike_limitation}
    \end{align}
    for all $t \ge 0$, where $\lambda_{\min}(A(\tau))^+ = \max\{\lambda_{\min}(A(\tau)), 0\}$. Then, $\xi$ is exponentially input-to-state stable (eISS) with respect to the inputs $D$. That is, $\exists \lambda > 0$ such that
    \begin{equation*}
        \|\xi(t+T)\| \le \exp(-\mu)\|\xi(t)\| + \exp(M-\mu)\int_t^{t+T} \|D(\tau)\|d\tau
    \end{equation*}
    for all $t \ge 0$. 
\end{lemma}

\begin{proof}
    See Appendix.
\end{proof}

In general, a priori verification of \eqref{eq: persistent_ex}-\eqref{eq: spike_limitation} for the obtained error dynamics \eqref{eq: error_compressed} is highly non-trivial. A sufficient condition is for $S$ to be positive-definite with a global upper bound on its minimum eigenvalue. While a global upper bound on the eigenvalues of $S$ corresponds to the very typical assumption of bounded state variables, together with a bounded parametric model for the mass and rate of change of the mass, it is expected in practice that $S$ will \textit{not} be positive-definite. In such a case, \eqref{eq: persistent_ex}-\eqref{eq: spike_limitation} are best regarded as conditions to be numerically verified in simulations prior to execution in real-world applications, or otherwise can be treated as a guide to path-planning. In the sequel, we make the latter notion explicit. 

Since $\lambda_{\min}$ is concave, a necessary condition to satisfy \eqref{eq: persistent_ex} is that
$$\lambda_{\min}\left(\int_t^{t+T} S(\tau)d\tau \right) \ge \mu$$
for all $t \ge 0$, which provides the three elementwise conditions:
\begin{align*}
    &\int_t^{t+T} S_{11}(\tau)d\tau > 0, \quad \int_t^{t+T} S_{22}(\tau)d\tau > 0, \\
    &\quad \int_t^{t+T} \left(S_{11}(\tau)S_{22}(\tau) - S_{12}(\tau)^2\right) d\tau > 0
\end{align*}

Suppose the mass model is uniformly bounded in the sense that there exist constants
\(
0<\underline a\le \overline a,\;
0<\underline b\le \overline b,\;
\overline c>0
\)
such that, for all admissible estimates $\hat\theta$ and all $t\ge 0$,
\begin{equation}\label{eq:grad_bounds}
\begin{split}
    \underline a \le \|\nabla_\theta m_{\hat\theta}(t)&\|_{K} \le \overline a,
\qquad
\underline b \le \|\nabla_\theta M_{\hat\theta}(t)\|_{K} \le \overline b, \\
&\big|\langle \nabla_\theta m_{\hat\theta}(t),\nabla_\theta M_{\hat\theta}(t)\rangle_{K}\big|\le \overline c,
\end{split}
\end{equation}
Then the diagonal entries of $S(t)$ admit the lower bounds
\begin{equation}\label{eq:Sdiag_lower}
\begin{split}
  &S_{11}(t)\ge \underline a^{\,2}\|w_L(t)\|^2-\overline c\,|w_L(t)^\top v_L(t)| \\ 
&S_{22}(t)\ge \underline b^{\,2}\|v_L(t)\|^2-\overline c\,|w_L(t)^\top v_L(t)|
\end{split}
\end{equation}
Hence, a convenient sufficient condition for the windowed diagonal positivity
$\int_t^{t+T}S_{ii}(\tau)\,d\tau>0$ is that, for all $t\ge 0$,
\begin{equation}\label{eq:diag_sufficient}
\begin{split}
    \int_t^{t+T}\Big(\underline a^{\,2}\|w_L(\tau)\|^2-\overline c\,|w_L(\tau)^\top v_L(\tau)|\Big)\,d\tau &>0 \\
    \int_t^{t+T}\Big(\underline b^{\,2}\|v_L(\tau)\|^2-\overline c\,|w_L(\tau)^\top v_L(\tau)|\Big)\,d\tau &>0.
\end{split}
\end{equation}
Condition \eqref{eq:diag_sufficient} can be understood as requiring that, on average, the vectors
$v_L$ and $w_L:=\dot v_L+ge_3$ are not nearly collinear on windows of the form $[t,t+T]$.

Alternatively, recognizing that
\begin{equation}\label{eq:energy_rate_identity}
w_L^\top v_L
=
(\dot v_L+ge_3)^\top v_L
=
\frac{d}{dt}\Big(\tfrac12\|v_L\|^2+ge_3^\top x_L\Big)
=:\dot H,
\end{equation}
we may interpret \eqref{eq:diag_sufficient} in terms of a windowed energy-drift constraint.
Indeed, by \eqref{eq:energy_rate_identity},
\begin{align*}
|H(t+T)-H(t)|&\le \int_t^{t+T}|\dot H(\tau)|\,d\tau \\
&= \int_t^{t+T}|w_L(\tau)^\top v_L(\tau)|\,d\tau,
\end{align*}
and thus \eqref{eq:diag_sufficient} implies, for all $t \ge 0$, the sufficient condition
\begin{equation}\label{eq:energy_drift_sufficient}
\begin{split}
    |H(t+T)-H(t)|
&<
\frac{1}{\overline c}\min\Bigg\{
\underline a^{\,2}\int_t^{t+T}\|w_L(\tau)\|^2\,d\tau,\; \\
&\qquad\qquad\underline b^{\,2}\int_t^{t+T}\|v_L(\tau)\|^2\,d\tau
\Bigg\}.
\end{split}
\end{equation}

Fix a horizon $T>0$ and a tension parameter $\tau\ge 0$.
On each window $[t,t+T]$, a \emph{spline in tension} is obtained by solving the variational problem
\begin{equation}\label{eq:tension_spline_action}
\min_{x\in\Omega}\int_t^{t+T}\Big(\|\ddot x(\sigma)+ge_3\|^2+\tau\|\dot x(\sigma)\|^2\Big)\,d\sigma,
\end{equation}
where $\Omega$ denotes the space of $C^1$ curves, piecewise $C^2$, satisfying boundary
conditions in $(x,\dot x)$ at $\sigma=t$ and $\sigma=t+T$.
For clamped endpoints, the gravity offset contributes only a boundary term in the action, and the interior Euler--Lagrange equation associated with \eqref{eq:tension_spline_action} reduces to 
\begin{equation}\label{eq:tension_spline_ode}
x^{(4)}(\sigma)- \tau x^{(2)}(\sigma)=0,\qquad \sigma\in[t,t+T].
\end{equation}
Thus, for $\tau>0$ the optimal segments are hyperbolic curves, while
for $\tau=0$ they reduce to cubic polynomials.
In the present context, \eqref{eq:tension_spline_action} is a natural smoothness surrogate, but it
is also important to note that it \emph{penalizes} precisely the quantities $\|w_L\|^2$ and
$\|v_L\|^2$ that appear in the sufficient excitation inequalities
\eqref{eq:diag_sufficient}--\eqref{eq:energy_drift_sufficient}. Hence, for fixed boundary
conditions, tension (and cubic) splines can be regarded as a conservative baseline from the point of view of identification, in the sense that the resulting trajectories provide limited excitation for mass estimation.

On the other hand, consider the quadratic net-thrust optimal control problem on $[t,t+T]$,
\begin{equation}\label{eq:quad_effort_ocp}
\min_{\mu(\cdot)}\int_t^{t+T}\|\mu(\sigma)\|^2\,d\sigma 
\end{equation} subject to $\mu(\sigma)=m_L(\sigma)\big(\ddot x_L(\sigma)+ge_3\big)+\dot m_L(\sigma)\dot x_L(\sigma)$
together with the same boundary conditions in $(x_L,\dot x_L)$.
In the constant-mass case $m_L(\sigma)\equiv m$, \eqref{eq:quad_effort_ocp} reduces (up to a
constant factor) to minimizing $\int\|\ddot x_L+ge_3\|^2$, whose stationary curves are cubic segments (and cubic splines in the presence of knots).
Moreover, introducing the relative mass-rate $\alpha(\sigma):=\dot m_L(\sigma)/m_L(\sigma)$,
if $\alpha(\sigma)$ remains approximately constant over a sufficiently short window, then
\eqref{eq:quad_effort_ocp} reduces (up to boundary terms) to the tension-spline action
\[
\int_t^{t+T}\Big(\|\ddot x_L(\sigma)+ge_3\|^2+\alpha(\sigma)^2\|\dot x_L(\sigma)\|^2\Big)\,d\sigma,
\]
so that the minimizers are approximated, to first order in $\sup_{\sigma\in[t,t+T]}|\dot\alpha(\sigma)|\,T$,
by a spline in tension with $\tau=\alpha^2$.
Taken together, these observations highlight a basic path-planning trade-off:
minimum-cost references tend to suppress $\|v_L\|$ and $\|w_L\|$, while persistent excitation
demands nontrivial variation in these signals for mass estimation.

A pragmatic compromise is therefore to use a smooth nominal reference (e.g., a tension spline on each window) and superimpose a small, bandwidth-limited excitation signal (i.e. a dither) chosen to satisfy the windowed inequalities \eqref{eq:diag_sufficient} (or the energy-drift surrogate \eqref{eq:energy_drift_sufficient}) while respecting actuator limits. This preserves near-minimum effort tracking in the nominal sense, but intentionally injects sufficient variability in the velocity and acceleration fields to support convergence of the mass estimator.

\subsubsection{Constant Mass Model}

Consider the special case where the mass is assumed to be (piecewise) constant, corresponding to a fluid filled cargo with discrete discharge times. That is, the model is simply given by $m_\theta(t) = \theta$, so that $M_\theta \equiv 0$ and $\displaystyle \frac{\partial m_\theta}{\partial \theta} = 1$. We denote the parameter estimate by $\hat{\theta} = \hat{m}$ for convenience. From Equation \eqref{eq: theta_regressor}, the mass estimation is updated via $\displaystyle{\dot{\hat{m}}(t) = \gamma w_L^T \left(\sum_{j=1}^4 \mu_j - \hat{m}w_L \right)}$
which yields the error dynamics $\dot{e}_m = -\gamma \|w_L\|^2 e_m + \gamma w_L^T \Delta_m$.
Lemma \ref{lemma: persistent_excitation} implies the following result.
\begin{proposition}
    Suppose that the load acceleration is bounded uniformly bounded, and  that there exist positive constants $\mu, T > 0$ such that
    \begin{equation}\label{eq: PE_constant_mass}
        \int_t^{t+T} \|\ddot{x}_L(t) + ge_3\|^2 dt \ge \mu
    \end{equation}
    for all $t \ge 0$. Then the error dynamics are eISS with respect to the input $\Delta_m$.
\end{proposition}

The contrast between mass estimation and optimal control is even sharper: cubic splines \emph{exactly} minimize the $L^2$ control effort \eqref{eq:quad_effort_ocp}, and thus minimize the excitation term in \eqref{eq: PE_constant_mass}. Accordingly, the same planning trade-off applies, but we select cubic-spline references augmented with a small dither to induce nontrivial velocity and acceleration variations.

\subsection{Inertia Estimation}
Next, we model the load’s time-varying inertia tensor $\mathbb{J}_L(t)$. Unlike mass, online inertia identification is generally impractical due to the high-dimensional PDE nature of internal fluid motion. Instead, using known tank geometry and a hydrostatic approximation, we precompute the tank–fluid inertia offline as a function of fill level and load-frame gravity direction, store it in a lookup table, and query it online using the estimated mass and measured attitude. Thus, inertia “estimation” is treated as a surrogate that reproduces the load rotational dynamics, and we specify operating and trajectory conditions under which the hydrostatic model remains valid and robust to unmodeled sloshing.

\subsection*{From Geometry to Inertia}
Suppose that the empty tank has mass $m_\mathcal{T}$ and the shape of the cavity $\mathcal{T} \subset \R^3$ is known, and has volume $V_{\mathcal{T}} := \text{Vol}(\mathcal{T})$. The tank is to be filled (either partially or entirely) with a fluid of known uniform density $\rho$. From the previous section, we may assume that for each $t \ge 0$, we have an estimate $\hat{m}(t)$ for the mass of the load (i.e. the combined mass of the tank and the fluid). Let $\sigma(t) \in [0,1]$ denote the proportion of the tank that is filled with fluid at time $t \ge 0$, which we call the \textit{fill level}. We see that
\begin{equation}\label{eq: fill_level}
    \sigma(t) = \frac{m(t) - m_{\mathcal{T}}}{\rho V_{\mathcal{T}}}
\end{equation}
for all $t \ge 0$. Given an estimate $\hat{m}(t)$ for the mass of the tank (for example, using the techniques described in the previous section), we similarly obtain an estimate $\hat{\sigma}(t)$ for the fill level using equation \eqref{eq: fill_level}. Moreover, the volume of the fluid at time $t \ge 0$ is given by $V_{\mathcal{F}}(t) = \sigma(t)V_{\mathcal{T}}$. To obtain an estimate for the moment of inertia, we will additionally need to know the region $\mathcal{F}(t) \subset \mathcal{T}$ that the fluid occupies at time $t \ge 0$. Accounting for nonlinear effects (sloshing, nonlinear waves, etc.) would be extraordinarily complicated, so we assume that the tank is sufficiently stationary that hydrostatic equilibrium is (approximately) attained. Later, we will discuss path-planning conditions that make this assumption reasonable in practice. In such a case, the surface of the fluid will look like a plane normal to the direction of gravity in the body frame, which is given by $g_L(t) = -R_L(t)^T e_3$. The fluid region $\mathcal{F}(t)$ is of the form:
\begin{equation}\label{eq: fluid_region}
    \mathcal{F}_h(t) = \mathcal{T} \cap \{x \in \R^3 \ \vert \ g_L(t)^T x \le h(t)\}
\end{equation}
for some $h(t) \in \R$ defined by the relation
$$\int_{\mathcal{F}_h(t)}dV =  \text{Vol}(\mathcal{F}_h(t)) = V_{\mathcal{F}}(t) = \sigma(t)V_{\mathcal{T}}$$
For fixed $t \ge 0$, the function $h \mapsto \text{Vol}(\mathcal{F}_h(t))$ is continuous, monotonically increasing, and satisfies
$$0 \le \text{Vol}(\mathcal{F}_h(t)) \le V_{\mathcal{T}}, \qquad \forall h \in \R$$
Hence, there exists a unique parameter $h(t)  = h(t)^\ast$ satisfying \eqref{eq: fluid_region} for all $t \ge 0$, which can be numerically approximated by a simple bisector method. Fix some reference point $\mathcal{O}$ on the tank (for example, the geometric center) in the body frame. Then, the moment of inertia tensors of the empty tank and the fluid with respect to $\mathcal{O}$ are given by: 
\begin{align*}
    \J_{\mathcal{T}, \mathcal{O}} &= \frac{m_{\mathcal{T}}}{V_{\mathcal{T}}} \int_{\mathcal{T}} (\|x\|^2 I_3 - xx^T)dV \\
    \J_{\mathcal{F}, \mathcal{O}}(t) &= \rho \int_{\mathcal{F}(t)} (\|x\|^2 I - xx^T)dV
\end{align*}
respectively, where $x$ denotes the position vector of a point with respect to the origin $\mathcal{O}$. The inertia tensor that appears in the dynamics is taken with respect to the center of mass of the load, which can be calculation via:
$$\mathcal{O}_{\text{cm}}(t) = \frac{m_{\mathcal{T}}}{\hat{m}(t)V_{\mathcal{T}}} \int_{\mathcal{T}} xdV + \frac{\hat{m}(t) - m_{\mathcal{T}}}{\hat{m}(t)V_{\mathcal{F}}(t)} \int_{\mathcal{F}(t)} xdV$$
as long as $V_{\mathcal{F}}(t) \ne 0$ (if not, we simply work with the tank's inertia directly).
From the parallel-axis theorem, the inertia tensor of the load (with respect to its center of mass), is given by
\begin{equation*}
\begin{aligned}
\J_L(t) &= \J_{\mathcal{T}, \mathcal{O}} + \J_{\mathcal{F}, \mathcal{O}}(t) \\
&\quad - \hat{m}(t)\Big(\|\mathcal{O}_{\mathrm{cm}}(t)\|^2 I_3
- \mathcal{O}_{\mathrm{cm}}(t)\mathcal{O}_{\mathrm{cm}}(t)^T\Big).
\end{aligned}
\end{equation*}

Such a calculation may be expensive to perform online, so we propose to perform the calculations offline for a sufficiently rich sampling of potential orientations $g_L \in S^2$ and fill fractions $\sigma$, and subsequently create a look-up table that can cheaply and easily be consulted in real-time based on estimates for the mass $\hat{m}(t)$ and orientation $R_L(t)$ of the load.

\subsection*{Fluid dynamics and validity of the hydrostatic approximation}

To characterize when the fluid region $\mathcal{F}(t)$ defined in \eqref{eq: fluid_region} provides a reasonable approximation of the true fluid configuration, we briefly review the governing fluid dynamics.

In the inertial frame, the motion of an incompressible Newtonian fluid of constant density $\rho$ and kinematic viscosity $\nu$ is governed by the Navier--Stokes equations
\begin{equation}
\label{eq:NS_inertial}
\frac{\partial v}{\partial t} + (v \cdot \nabla)v
=
-\frac{1}{\rho}\nabla P
+ \nu \Delta v
- ge_3,
\end{equation}
where $v(x,t)$ and $P(x,t)$ denote the velocity and pressure fields, respectively. Consider the body-frame coordinates $r \in \mathbb{R}^3$ by $x = x_L(t) + R_L(t) r$. Define the velocity of the fluid relative to the tank, expressed in body coordinates, as
\begin{equation}
\label{eq:relative_velocity}
u(r,t)
=
R_L(t)^\top\big(v(x,t) - \dot{x}_L(t)\big)
-
\Omega_L(t) \times r,
\end{equation}
where $\Omega_L(t)$ satisfies $\dot{R}_L = R_L \hat{\Omega}_L$. Expressed in the rotating and translating body frame, the Navier--Stokes equations take the standard form
\begin{align}
\label{eq:NS_body}
\frac{\partial u}{\partial t} + (u \cdot \nabla)u
=&
-\frac{1}{\rho}\nabla \tilde{P}
+ \nu \Delta u
+ g_L
- a_L
- \dot{\Omega}_L \times r \nonumber\\
& - \Omega_L \times (\Omega_L \times r)
- 2\,\Omega_L \times u,
\end{align}
where all quantities are expressed in body coordinates:
\[
\tilde{P}(r,t) = P(x_L(t) + R_L(t)r,t), \qquad
a_L = R_L^\top \ddot{x}_L,
\]
and the differential operators are taken with respect to $r$. 

Hydrostatic equilibrium corresponds to the fluid being at rest relative to the tank, i.e., $u(r,t) \equiv 0$.

In this case, viscous stresses vanish identically and the momentum balance reduces to
\begin{equation}
\label{eq:hydrostatic_balance}
\nabla \tilde{P}
=
\rho\Big(
g_L
- a_L
- \dot{\Omega}_L \times r
- \Omega_L \times (\Omega_L \times r)
\Big).
\end{equation}
Consequently, surfaces of constant pressure are orthogonal to the effective acceleration field acting on the fluid.

If the translational acceleration $a_L$, angular velocity $\Omega_L$, and angular acceleration $\dot{\Omega}_L$ are sufficiently small, the dominant contribution to \eqref{eq:hydrostatic_balance} is gravity, and the free surface of the fluid is well-approximated by a plane orthogonal to $g_L(t) = -R_L(t)^\top e_3$, as we previously assumed in \eqref{eq: fluid_region}.

Departures from hydrostatic equilibrium occur when the effective forcing in \eqref{eq:hydrostatic_balance} varies rapidly. In particular, large translational jerk $\dddot{x}_L$, angular acceleration $\dot{\Omega}_L$, or angular jerk $\ddot{\Omega}_L$ can excite internal fluid motion (sloshing) through the unsteady and convective terms of the body-frame Navier--Stokes equations \eqref{eq:NS_body}. From a frequency-domain viewpoint, this happens when the forcing contains energy near the dominant sloshing modes, i.e., when the forcing varies on time scales comparable to or faster than the fluid response. Thus, to promote validity of the hydrostatic region model \eqref{eq: fluid_region}, we impose trajectory-smoothness conditions that limit the magnitude and bandwidth of $\dddot{x}_L$ and $\ddot{\Omega}_L$, while keeping $a_L$ and $\dot{\Omega}_L$ moderate so gravity remains dominant in \eqref{eq:hydrostatic_balance}.

This motivates smooth, low-jerk reference trajectories for the load. In $\R^3$, minimum-jerk curves (minimizers of $\int_0^T |\dddot{x}(t)|^2dt$ under standard boundary conditions) are quintic polynomials, and the concept extends to Riemannian manifolds via minimum-jerk variational curves defined through third-order covariant derivatives \cite{camarinha1995splines}. In our setting, such trajectories reduce excitation of unmodeled fluid dynamics and make the look-up-table inertia strategy meaningful.

We (i) generate smooth, low-jerk load references using (Riemannian) splines (e.g., cubics or splines in tension), (ii) inject small low-frequency dithers when mass is unmeasured to ensure sufficient variability for mass convergence, (iii) estimate the load mass $\hat m(t)$ and attitude $R_L(t)$ online, and (iv) update inertia via an offline look-up table from \eqref{eq: fluid_region} indexed by $(\hat m(t),R_L(t))$ (equivalently by fill fraction $\hat\sigma(t)$ and body-frame gravity direction $g_L(t)=-R_L(t)^\top e_3$). Steps (i)–(ii) are path-planning choices that enforce the regularity needed for (iii)–(iv): they limit rapid forcing variations to mitigate sloshing, yet retain enough excitation to identify mass in the absence of direct measurements.

\section{Simulation Results}

In this section we present numerical simulations of the cooperative aerial load transportation system composed of four quadrotors connected to a rigid load via massless cables.  The simulation horizon is $15$~s and all initial conditions are slightly perturbed from the nominal hovering configuration.

The system state includes the load translational and rotational states $(x_L,v_L,R_L,\Omega_L)$, as well as the direction $q_j\in S^2$ and angular velocity $w_j$ of each cable. The load mass is either constant or time-varying depending on the scenario, and an online estimator $\hat m(t)$ is used by the controller for gravity compensation. The inertia tensor $\J_L(m_L)$ is updated consistently. Sensor measurements for position, velocity, attitude, and angular velocity include small deterministic noise terms. These are implemented as smooth sinusoidal perturbations with amplitudes of $0.01$--$0.02$ in position and velocity, and $0.005$ in attitude, consistent with typical low-amplitude sensor disturbances.  
Furthermore, an external disturbance force emulating wind is optionally applied.  
The wind disturbance is modeled as a smooth, time-varying signal of the form
$F_{\mathrm{wind}}(t) = 0.3
\begin{bmatrix}
\sin(0.4 t) &
\cos(0.6 t) &
\sin(0.8 t)
\end{bmatrix} \! \text{N},
$
which acts directly on the translational dynamics of the load.  
These perturbations are representative of environmental disturbances encountered in outdoor flight and allow evaluating the robustness of the proposed control architecture.  

Two operating modes are considered: in the constant-mass case, the load mass remains fixed during the simulation; in the variable-mass case, the load mass decreases exponentially, which modifies both $m_L(t)$ and the inertia tensor $\J_L(t)$, and the estimator compensates for this variation using the measured acceleration. Fig.~\ref{fig:position} illustrates the evolution of the load position components for both mass scenarios. For the constant-mass case the trajectory converges smoothly to the desired equilibrium at the origin.  
When the mass varies, the initial mismatch between $m_L(0)$ and $\hat m(0)$ produces a larger transient, but the controller ultimately stabilizes the load once the mass estimate converges.

\begin{figure}[h!]
    \centering
    \includegraphics[width=\linewidth]{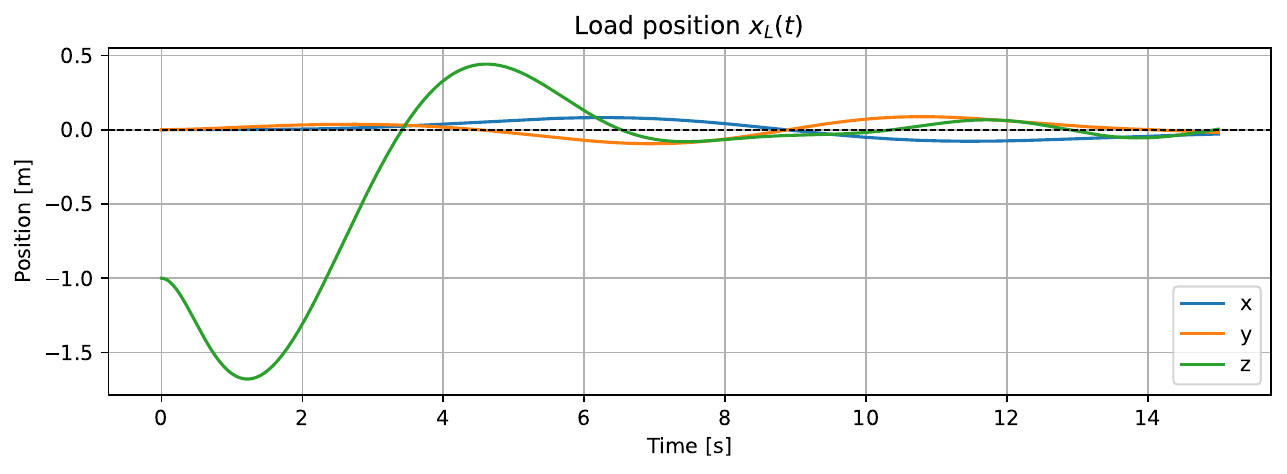}\\[4pt]
    \includegraphics[width=\linewidth]{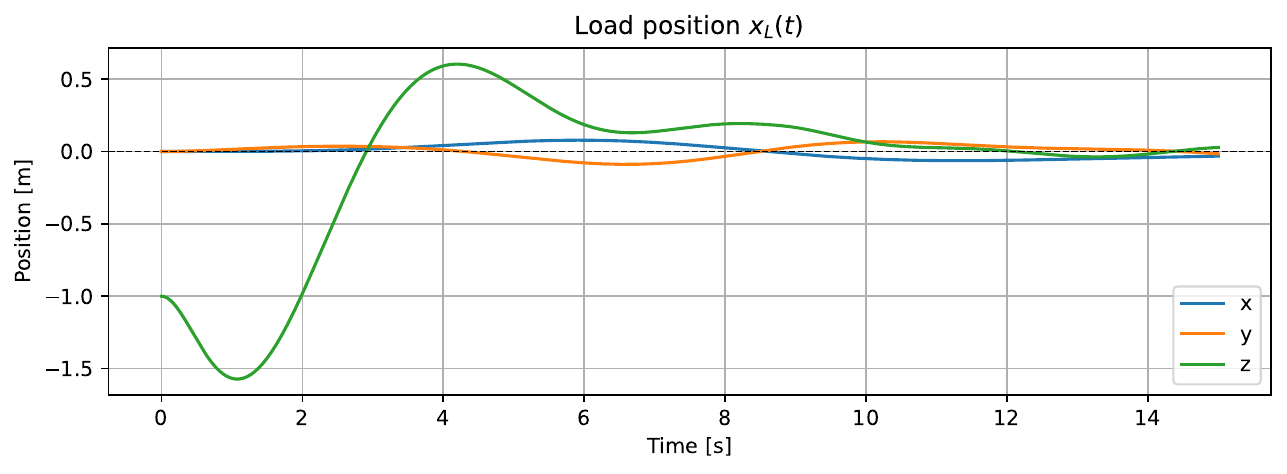}
    \caption{Load position trajectory $x_L(t)$ under (top) constant mass and (bottom) exponentially decaying mass.}
    \label{fig:position}
\end{figure}

Fig.~\ref{fig:mass} compares the true mass $m_L(t)$ with the online estimate $\hat m(t)$ produced by the adaptive law.  
In the constant-mass case the estimator converges rapidly despite the large initial error.  
For the variable-mass case, the estimator successfully tracks the exponential decay, enabling accurate gravity compensation throughout the simulation.

\begin{figure}[t]
    \centering
    \includegraphics[width=\linewidth]{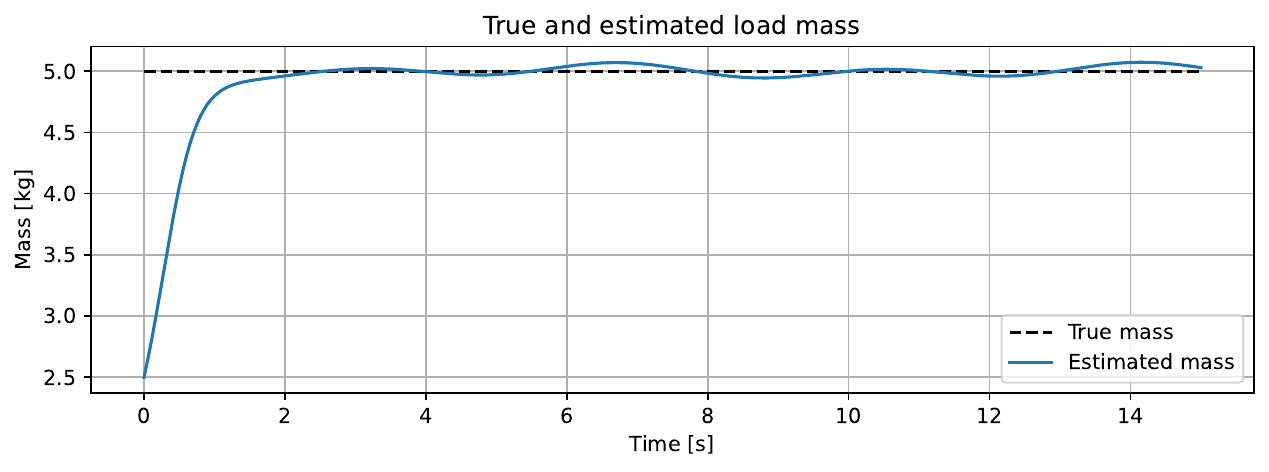}\\[4pt]
    \includegraphics[width=\linewidth]{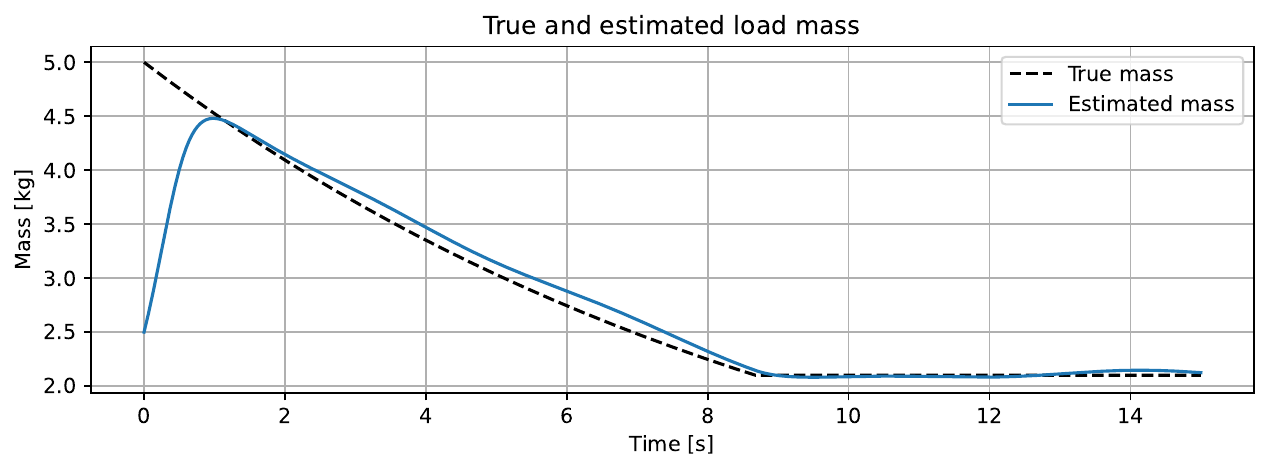}
    \caption{True and estimated load mass for the two simulation scenarios.}
    \label{fig:mass}
\end{figure}

The diagonal entries of the load inertia tensor, together with their estimates, are displayed in Fig.~\ref{fig:inertia}.  
When the mass is constant, the inertia remains fixed and the estimator produces a matching constant value.  
In the variable-mass scenario, the inertia decreases smoothly according to the mass decay model, and the estimate derived from $\hat m(t)$ accurately captures this evolution.

\begin{figure}[t]
    \centering
    \includegraphics[width=\linewidth]{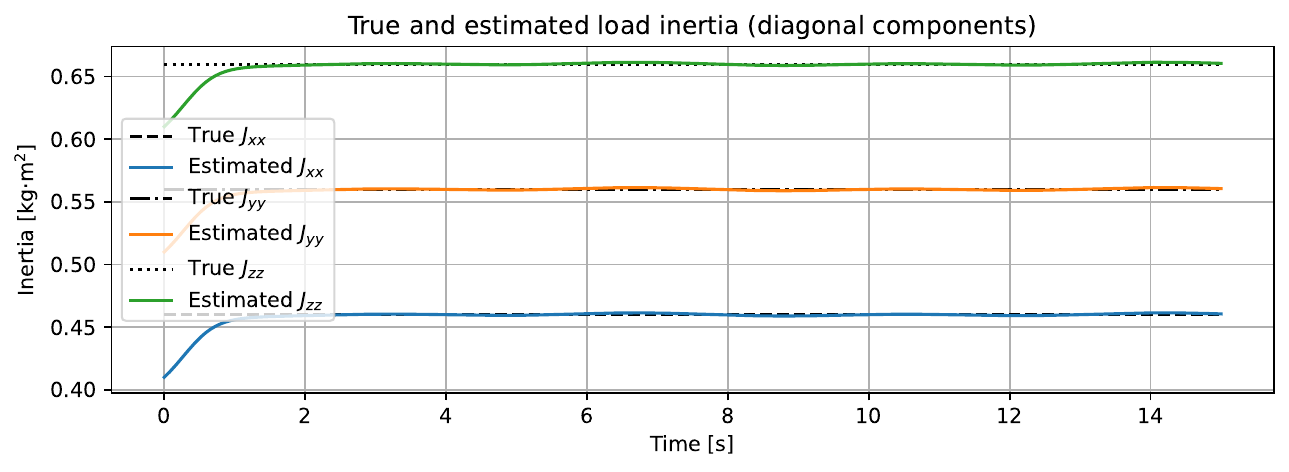}\\[4pt]
    \includegraphics[width=\linewidth]{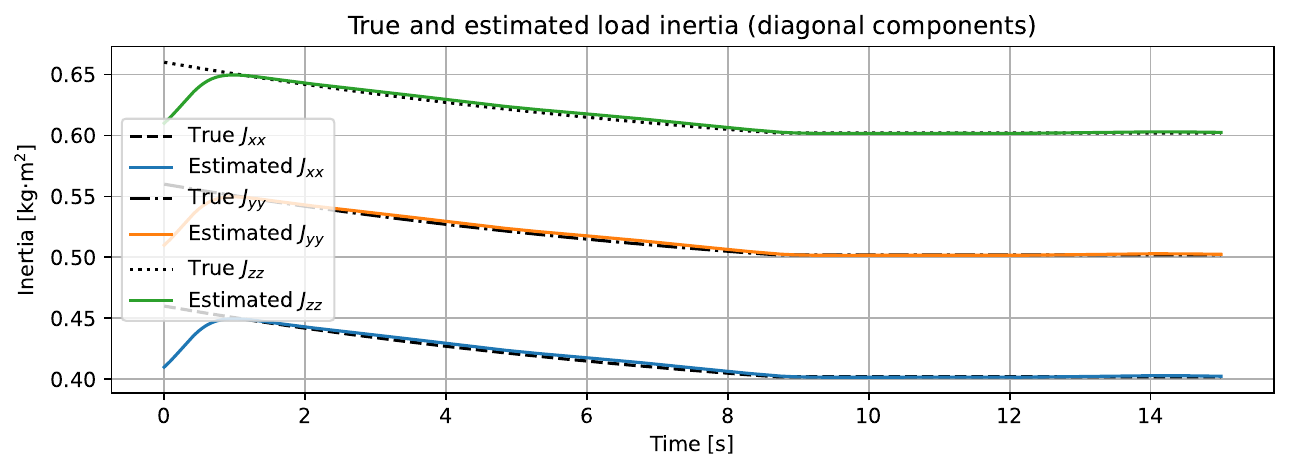}
    \caption{Diagonal elements of the inertia tensor $J_L(t)$ (true vs.\ estimated) for both mass models.}
    \label{fig:inertia}
\end{figure}



\section*{Conclusions}
We studied load-parameter estimation for cooperative aerial transportation when a rigid payload exhibits time-varying mass and inertia, as in fluid-carrying operations with evolving internal mass distribution. Building on an intrinsic geometric model of a multi-quadrotor--load system with elastic tethers, we developed an online mass estimator driven by standard state estimates and commanded inputs, and proposed a hydrostatic, geometry-informed inertia surrogate implemented via an offline look-up table indexed by the estimated fill level and measured load attitude. Numerical simulations show accurate mass adaptation and effective inertia updates, yielding robust closed-loop performance under variable payload properties. Future work will focus on multi-UAV experimental validation, including improved onboard state estimation for load acceleration/attitude, trajectory-design guidelines balancing persistent excitation with low-jerk motion, and extensions beyond the hydrostatic regime via sloshing-aware reduced-order models or data-driven residual corrections under actuation limits, delays, and unmodeled aerodynamics.

\section{Appendix}
\textit{Proof of Lemma \ref{lemma: persistent_excitation}}: By direct calculation,
    $$\frac{d}{dt} \|\dot{\xi}\| = \frac{-\xi^T A \xi + \xi^T D}{\|\xi\|} \le -\lambda_{\min}(A)\|\xi\| + \|D\|$$
    Since $A$ is symmetric. Using Gronwall's inequality on the interval $[t, t+T]$, we see that
    \begin{align*}
        \|\xi(t+T)\| &\le \exp\left(-\int_t^{t+T} \lambda_{\min}(A(\tau))d\tau \right)\|\xi(t)\| \\
        &+ \int_t^{t+T} \exp\left( -\int_s^{t+T}\lambda_{\min}(A(\tau))d\tau\right)\|D(s)\|ds
    \end{align*}
    Next, we observe that for all $s \in [t, t+T]$,
    \begin{align*}
        \int_s^{t+T} \lambda_{\min}(A(\tau))d\tau &= \int_t^{t+T} \lambda_{\min}(A(\tau))d\tau - \int_t^s \lambda_{\min}(A(\tau))d\tau \\
        &\ge \mu - \int_t^s \lambda_{\min}(A(\tau))^+d\tau \\
        &\ge \mu - \int_t^{t+T} \lambda_{\min}(A(\tau))^+ d\tau \\
        &\ge \mu - M
    \end{align*}
    Hence, for all $t \ge 0$,
    $$\|\xi(t+T)\| \le \exp(-\mu)\|\xi(t)\| + \exp(M-\mu)\int_t^{t+T} \|D(\tau)\|d\tau.$$


\begin{thebibliography}{99}

\bibitem{pacheco1} V. P. Bacheti, D. K. D. Villa, R. Lozano, M. Sarcinelli-Filho, and P. Castillo,
“Tension-Based Reconfigurable Multi-Agent Formation for Aerial Load Transportation,”
IEEE Access, vol. 13, pp. 114098–114116, 2025.



\bibitem{camarinha1995splines}
M.~Camarinha, F.~Silva~Leite, and P.~Crouch,
``Splines of class $C^k$ on non-Euclidean spaces,''
\emph{IMA Journal of Mathematical Control and Information},
vol.~12, no.~4, pp.~399--410, 1995.


\bibitem{pacheco2} E. D. S. Cardoso, V. P. Bacheti, and M. Sarcinelli-Filho,
“A Leader-Follower Control System for a Package-Delivery Drone,”
IEEE Access, vol. 13, pp. 55313–55323, 2025.


\bibitem{Fink} M. Fink, N. Michael, S. Kim, and V. Kumar, “Planning and control for cooperative manipulation and transportation with aerial robots,” Int. J. Robotics Research, vol. 30, no. 3, pp. 324–334, 2011.



\bibitem{jacob1} Goodman, J. R., Beckers, T., $\&$ Colombo, L. J., “Geometric Control for Load Transportation with Quadrotor UAVs by Elastic Cables”, IEEE Transactions on Control Systems Technology, vol. 31, no. 6, pp. 2848-2862, 2023.

\bibitem{jacob2} Goodman, J. R. $\&$ Colombo, L. (2022). Geometric Control of Two Quadrotors Carrying a Rigid Rod with Elastic Cables. Journal of Nonlinear Science, 32:65.

\bibitem{Elastic}
P. Kotaru, G. Wu and K. Sreenath, "Dynamics and control of a quadrotor with a payload suspended through an elastic cable," 2017 American Control Conference (ACC), 2017, 3906-3913.
\bibitem{Lee-TCST}
T. Lee. Geometric control of quadrotor uavs transporting a cablesuspended rigid body. IEEE Transactions on Control Systems Technology, vol. 26, no. 1, pp. 255-264, 2018.

\bibitem{lee_geometric_2010} T. Lee, Taeyoung, M. Leok and N. H. {McClamroch}, Geometric tracking control of a quadrotor {UAV} on {SE}(3). 49th {IEEE} Conference on Decision and Control ({CDC}), 5420--5452, 0191-2216, 2010.

\bibitem{LeeSrePICDC13}
T.~Lee, K.~Sreenath, and V.~Kumar, Geometric control of cooperating multiple
  quadrotor {UAV}s with a suspended load.  \textit{Proceedings of the IEEE
  Conference on Decision and Control}, 2013, pp. 5510--5515.

\bibitem{loiano} G. Li, R. Ge, G. Loianno. Cooperative transportation of cable suspended payloads with mavs using monocular vision and inertial sensing. IEEE Robotics and Automation Letters, 6(3), 5316-5323, 2021.

\bibitem{MR} J. E. Marsden, T. S. Ratiu, Introduction to Mechanics and Symmetry: A Basic Exposition of Classical Mechanical System. Springer, 1999. 


\bibitem{petitti2020inertial}
A.~Petitti, D.~Sanalitro, M.~Tognon, A.~Milella, J.~Cort\'es, and A.~Franchi,
``Inertial Estimation and Energy-Efficient Control of a Cable-suspended Load with a Team of UAVs,''
in \emph{Proc. Int. Conf. on Unmanned Aircraft Systems (ICUAS)},
2020.

\bibitem{SreLeePICDC13}
K.~Sreenath, T.~Lee, and V.~Kumar, ``Geometric control and differential
  flatness of a quadrotor {UAV} with a cable-suspended load,'' in
  \textit{IEEE Conference on Decision and Control}, 2013,
  2269--2274.

\bibitem{sun2025agilecoop}
S.~Sun, X.~Wang, D.~Sanalitro, A.~Franchi, M.~Tognon, and J.~Alonso-Mora,
``Agile and Cooperative Aerial Manipulation of a Cable-Suspended Load,''
\emph{Science Robotics}, 2025.




\bibitem{tognon2018internalforce}
M.~Tognon, C.~Gabellieri, L.~Pallottino, and A.~Franchi,
``Aerial Co-Manipulation With Cables: The Role of Internal Force for Equilibria, Stability, and Passivity,''
\emph{IEEE Robotics and Automation Letters},
vol.~3, no.~3, pp.~2577--2583, 2018.



\bibitem{wu} G. Wu and K. Sreenath, Geometric control of quadrotors transporting a rigid-body load, in IEEE Conference on Decision and Control, Los Angeles, CA, Dec. 2014, pp. 6141-6148.
\end{thebibliography}
\end{document}